\documentclass[submission,copyright,creativecommons]{eptcs}

\providecommand{\event}{GandALF 2026} 

\usepackage{iftex}
\usepackage[T1]{fontenc}

\ifpdf
\usepackage{underscore}
\else
\usepackage{breakurl}
\fi

\usepackage{tikz}
\usepackage{paralist}
\usepackage{amsmath,amssymb,amsthm,mathtools}
\usepackage{enumitem}
\usepackage{verbatim}
\usepackage{color}

\usepackage{cite}

\theoremstyle{definition}
\newtheorem{definition}{Definition}[section]
\newtheorem{example}[definition]{Example}
\newtheorem{remark}[definition]{Remark}

\theoremstyle{plain}
\newtheorem{lemma}[definition]{Lemma}
\newtheorem{proposition}[definition]{Proposition}
\newtheorem{theorem}[definition]{Theorem}
\newtheorem{corollary}[definition]{Corollary}

\newcommand{\Obs}{\mathsf{Obs}}

\newcommand{\Nat}{\mathbb{N}}

\newcommand{\Trc}{\mathop{\mathrm{Trc}}}
\newcommand{\Proc}{\mathop{\mathrm{Proc}}}

\newcommand{\q}[1]{``#1''}
\newcommand{\Pref}{\operatorname{Pref}}

\newcommand{\Sig}{\Sigma}

\newcommand{\act}{\Sig}

\newcommand{\true}{\mathsf{tt}}
\newcommand{\false}{\mathsf{ff}}
\newcommand{\HML}{\mathrm{HML}}

\title{A Topological Framework for Finite Behavioural \\ Observations and Verification}

\author{
Antonis Achilleos
\institute{Department of Computer Science,\
Reykjavik University,\
Reykjavik, Iceland}
\email{antonios@ru.is}
\and
Vasiliki Kyriakou
\institute{Department of Computer Science,\
Reykjavik University,\
Reykjavik, Iceland}
\email{vasilikikyr@ru.is}
}

\def\titlerunning{Topological Verification}
\def\authorrunning{Antonis Achilleos and Vasiliki Kyriakou}

\begin{document}

\maketitle

\begin{abstract}

Formal verification and monitorability are based on finite observations, which allow properties to be verified from finite information about system behaviour.
We study such observations through the topologies they generate on spaces of processes. We first consider trace-based topologies and show that finite trace observations on \(\Sigma^\omega\) induce the Cantor topology, while the topology corresponding to full trace inclusion is the discrete one. We then move to arbitrary process spaces, where finite trace observations define the topology \(\tau_O\), and show that simulation observations generate a strictly finer topology \(\tau_{\mathrm{sim}}\). Next, we prove a general verification theorem showing that, for any topology generated by finite observations, open sets are exactly the properties verifiable by those observations. We instantiate this result for \(\tau_O\) and \(\tau_{\mathrm{sim}}\), obtaining multi-trace and simulation monitorability as concrete cases. Finally, we examine the effect of replacing simulation with stronger relations, showing that finite-depth bisimulation induces a topology finer from \(\tau_{\mathrm{sim}}\).

\end{abstract}

\section{Introduction}

\textit{Reactive systems} are usually understood through their observable behaviour \cite{Aceto07,10.1007/978-3-642-82453}. \textit{Labelled transition systems (LTSs)} \cite{Tretmans2008,10.1007/3-540-44685,10.1007/3-540-53479} provide a standard model for describing such behaviour. In this model, states represent processes, and labelled transitions represent the actions that processes may perform. Different behavioural views arise depending on which aspects of process behaviour are observable. A trace-based observer records only the
sequences of actions that a process can execute, while a branching-time observer
also takes into account how choices are organized during the execution. The distinction between linear-time and branching-time behaviour is central in
process semantics \cite{4cd7021d283c4edab373d42f1d69a14c,GLABBEEK20013}.

This distinction is also important for \textit{verification} \cite{c9c054ed-1535-3e93-8569-48e7acb4c9c7,osti835920,KLEIJNEN1995145}. In many verification settings, one does not observe the whole behaviour of a system at once. Instead, one obtains finite observations and uses them to conclude whether the system satisfies a property. \textit{Runtime verification} \cite{inbook,LEUCKER2009293} and \textit{monitorability} \cite{Aceto2021,10.1145/3290365,inbook2} study finite observations from an operational point of view. Runtime verification is a lightweight verification technique in which the current execution of a system is analyzed with respect to a given specification. Monitorability formalizes when such finite observations are enough to reach a conclusive verdict. In this setting, a \textit{monitor} can be understood as a computational entity that reads a trace, namely a sequence of observable system events, and attempts to determine from a finite prefix whether the observed behaviour satisfies a given property or not \cite{Aceto2021}. A property
is \textit{satisfaction-monitorable} if its satisfaction can be concluded after observing
a finite prefix \cite{Aceto2021}. Thus, monitorability is
closely connected to the general question of which properties can be confirmed using finite information.

This question has a natural topological interpretation. Both monitorability and the \textit{Cantor topology} \cite{Moschovakis2009,Kechris1995} on infinite traces are based on finite prefixes. This connection suggests a broader question: how does the choice of finite observations affect the topology induced on the process space, and consequently which properties can be verified from these observations?

We address this question by studying several types of finite behavioural observations and the topologies they induce. Rather than starting from a particular monitor model, we start from the observations themselves. Each finite observation determines a set of processes satisfying that observation,
which can serve as a basic open set. When these observation sets form a basis, the resulting topology characterizes the properties that can be positively verified from the corresponding finite observations. Different observation
types retain different amounts of behavioural information, and comparing the topologies they induce allows us to determine which properties can be verified from finite observations.

Our approach therefore moves beyond the standard Cantor topology on individual traces and considers richer behavioural observations. We study finite trace observations, finite multi-trace observations, and finite process-based observations arising from simulation and finite-depth bisimulation. Finite trace observations record sequences of actions and capture linear behaviour. Finite multi-trace observations record finite
collections of traces, while process-based observations additionally capture branching behaviour. Finite-depth bisimulation observations compare the behaviour of processes in both directions up to a bounded depth.

The main contribution of this paper is a unified framework for constructing and comparing topologies induced by finite behavioural observations. 
The general correspondence between openness and verification 
has been observed and studied in previous work \cite{abramsky2011domaintheorylogicobservable,DIEKERT201429,CamerloDagnino2026,vickers1988topology} and it follows from the basis characterization of open sets.  We lean on this correspondence, with a focus on the finite observations as the intermediate notion that enables a system's verification and induces the corresponding topology. We instantiate this formulation with several concrete behavioural observation types and characterize the topologies they induce. We compare finite trace and simulation observations and establish that simulation observations yield a strictly finer topology compared to trace observations. We further investigate whether complete simulation and ready simulation observations form bases, and characterize the topology induced by finite-depth bisimulation observations.

In the present work, we develop these ideas in several steps. Section \ref{sec:trace-based topologies} begins with full trace inclusion and the Alexandroff topology. We then consider finite trace observations on \(\Sigma^\omega\), showing that they
induce the Cantor topology, while full trace inclusion collapses to the discrete topology.

Section \ref{sec:process-topologies} moves from spaces of traces to arbitrary process spaces. It defines the trace observation topology \(\tau_O\), induced by finite positive trace observations, and shows that it coincides with the
topology induced by finite-depth trace inclusion. The section then introduces the simulation topology \(\tau_{\mathrm{sim}}\), whose basic observations are finite loop-free processes and which captures finite branching behaviour. We show that $\tau_O\subsetneq\tau_{\mathrm{sim}}$.

Section~\ref{sec:verification} abstracts from these concrete constructions and establishes the general verification connection: whenever finite observation sets form a basis, the properties verifiable by finite observations are exactly the open sets of the induced topology. In Subsection~\ref{sec:monitorability}, this result applies to the two main topologies studied in the paper, yielding multi-trace monitorability for \(\tau_O\) and simulation monitorability for \(\tau_{\mathrm{sim}}\).

Subsection~\ref{sec:beyond simulation} examines whether the same construction extends to stronger behavioural relations. We show that finite complete-simulation and ready-simulation observations do not, in general, satisfy the covering property and therefore do not, in general, form bases, whereas finite-depth bisimulation observations do. The resulting topology
\(\tau^{\mathrm{fin}}_{\mathrm{bis}}\) differs from
\(\tau_{\mathrm{sim}}\). In particular, the set of deadlocked processes is clopen in \(\tau^{\mathrm{fin}}_{\mathrm{bis}}\) but is not open in \(\tau_{\mathrm{sim}}\). Finally, Subsection \ref{sec: modal logic} demonstrates that we can generate the same process topologies by using modal formulae as observations instead of finite loop-free processes.

\subsection*{Related Work}
\label{sec:related-work}

Our work lies at the intersection of finite behavioural observations, topology, behavioural semantics, and verification. Several existing works study closely related connections between these areas. We briefly review the most relevant ones and then clarify how our framework differs from these existing constructions.

\paragraph{Topology and finite observations.}
Vickers in \cite{vickers1988topology}, and Abramsky and Vickers in \cite{DBLP:journals/mscs/AbramskyV93} study topology as a way of representing information obtained through observations. In particular, Abramsky and Vickers develop this connection from an algebraic perspective, using structures such as monoids to model observations of processes. These works provide important foundations for relating finite observations to topological and logical notions. 
Our work shares this perspective, but starts from concrete finite behavioural observations of LTSs and studies the topologies induced by different observation types.

\paragraph{Topology and monitorability.} Diekert and Leucker \cite{DIEKERT201429} establish that open sets in the Cantor topology correspond to satisfaction-monitorable properties. Camerlo and Dagnino \cite{CamerloDagnino2026} extend this perspective by studying monitorable properties of higher topological complexity. In contrast, we move beyond individual traces and study how finite trace, multi-trace, simulation, and finite-depth bisimulation observations induce topologies on process spaces.

\paragraph{Behavioural observations.}
Simulation, complete simulation, ready simulation, and bisimulation are standard behavioural relations for LTSs~\cite{GLABBEEK20013}. Huth \cite{lmcs:2271} also studies finite behavioural approximations of LTSs from a topological perspective. Our work instead starts from specific finite observations based on traces, simulation, and finite-depth bisimulation, and investigates when the corresponding observation sets form bases and what topologies they induce. This allows us to compare the induced topologies, including the strict inclusion
\(\tau_O\subsetneq\tau_{\mathrm{sim}}\), and to identify cases where stronger behavioural observations do not form a basis. \\ \\
Overall, previous work provides the individual connections between observations, topology, monitorability, and behavioural semantics. Our work combines these perspectives by studying how different finite behavioural observations induce different topologies on process spaces and which properties can be verified from those observations.

\section{Preliminaries}
\label{preliminaries}

Let \(\Sigma\) be a finite alphabet. We write \(\Sigma^*\) and \(\Sigma^\omega\) for the sets of finite and infinite words over \(\Sigma\), respectively, and define $\Sigma^\infty=\Sigma^*\cup\Sigma^\omega$. The empty word is denoted by \(\varepsilon\). For \(s\in\Sigma^*\) and \(t\in\Sigma^\infty\), we write \(s \preceq t\) if \(s\) is a finite prefix of \(t\), possibly with \(s=t\). For \(t \in \Sigma^\infty\), we write $\Pref(t)=\{\,s\in\Sigma^*\mid s\preceq t\,\}$, for the set of finite prefixes of \(t\). For \(t\in\Sigma^\omega\) and \(k\in\mathbb N\), we write \(t\!\upharpoonright k\), for the prefix of \(t\) of length \(k\).

\paragraph{Labelled Transition Systems and Traces.}
A \textit{labelled transition system (LTS)} is defined by a triple
$L=\langle \Proc,\Sigma,\to\rangle$, where \(\Proc\) is a set of processes, \(\Sigma\) is a finite alphabet of actions, and $\to\ \subseteq\ \Proc\times\Sigma\times\Proc$ is the transition relation. We write \(p\xrightarrow{a}q\) when \((p,a,q)\in\to\).

For \(p\in\Proc\) and \(s=a_1\dots a_n\in\Sigma^*\), we say that \(p\) has \(s\) as a finite trace, if there exist processes $p_0,p_1,\ldots,p_n\in\Proc$, such that \(p_0=p\) and $p_{i-1}\xrightarrow{a_i}p_i$, for every $1\leq i\leq n$.
In particular, \(\varepsilon\) is a finite trace of every process. For \(p\in\Proc\) and \(t=a_1a_2\cdots\in\Sigma^\omega\), we say that \(p\) has \(t\) as an infinite trace if there exists a sequence of processes $p_0,p_1,p_2,\ldots\in\Proc$, such that \(p_0=p\) and $p_{i-1}\xrightarrow{a_i}p_i$, for every $i\geq 1$. For each process \(p\in\Proc\), we write $\Trc(p)\subseteq\Sigma^\infty$, for the set of finite and infinite traces executable from \(p\). Thus, \(s\in\Trc(p)\cap\Sigma^*\) means that \(p\) can execute the finite trace \(s\), while \(t\in\Trc(p)\cap\Sigma^\omega\) means that \(p\) admits the infinite trace \(t\). 

For a state \(p\in\Proc\), we write $I(p)=\{\,a\in\Sigma\mid \exists p'\in\Proc, \text{ such that }p\xrightarrow{a}p'\,\}$ for the set of actions initially enabled at \(p\). If \(I(p)=\varnothing\), then \(p\) is a deadlock, meaning that it has no outgoing transitions.

\begin{remark}
We use the term \textit{process} for an element of the state space of an LTS. Thus, throughout the paper, a process denotes an individual state of the considered LTS rather than the transition system as a whole.
\end{remark}

\paragraph{Finite Loop-Free Processes and Simulation.}
We use \textit{finite loop-free processes} as finite observations of process behaviour. Let \(\Proc_F\) be the set of terms generated by the grammar $q ::= 0 \mid a.q \mid q+q$, $a\in\Sigma$. For each action \(a\) and terms \(q,q'\), we write \(q\xrightarrow{a}_F q'\) iff
\emph{(i)}  
\(q=a.q'\), or
\emph{(ii)} \(q=q_1+q_2\), for some \(q_1,q_2\), and
\(q_1\xrightarrow{a}_F q'\) or \(q_2\xrightarrow{a}_F q'\) holds.
The process \(0\) has no outgoing transitions. Since the grammar contains no recursion or variables, every \(q\in\Proc_F\) is finite and loop-free.
We write $LTS_F=\langle \Proc_F,\Sigma,\to_F\rangle$ for the LTS of finite loop-free processes.

Let $L_1=\langle\Proc_1,\Sigma,\to_1\rangle$ and $L_2=\langle\Proc_2,\Sigma,\to_2\rangle$ be LTSs over the same alphabet. A relation $R\subseteq\Proc_1\times\Proc_2$ is a \emph{simulation} from $L_1$ to $L_2$ if, whenever $(p_1,p_2)\in R$ and $p_1\xrightarrow{a}_1 p_1'$, there exists $p_2'\in\Proc_2$, such that $p_2\xrightarrow{a}_2 p_2'$ and $(p_1',p_2')\in R$. For $p_1\in\Proc_1$ and $p_2\in\Proc_2$, we write $p_1\sqsubseteq_S p_2$, if there exists a simulation $R\subseteq\Proc_1\times\Proc_2$, such that $(p_1,p_2)\in R$. In particular, for $q\in\Proc_F$ and $p\in\Proc$, $q\sqsubseteq_S p$ means that the process $p$ simulates the finite loop-free observation $q$.

In the following, we fix an LTS $L=\langle \Proc,\Sigma,\to\rangle$ that, without loss of generality, includes the processes of $LTS_F$.

\paragraph{Topological Notation.} For a set \(X\), we write \(\mathcal P(X)\) for its powerset and \(\mathcal P_f(X)\) for the set of its finite subsets. If \(\mathcal F\) is a family of subsets of \(X\), then $\langle \mathcal F\rangle$ denotes the topology generated by \(\mathcal F\), that is, the smallest topology
on \(X\) containing every set in \(\mathcal F\).

A family \(\mathcal B\subseteq\mathcal P(X)\) is a basis for a topology \(\tau\) on \(X\) if \(\mathcal B\subseteq \tau\) and every open set is a union of elements of \(\mathcal B\). 
The following lemma characterizes a topology's open sets with respect to its basis.

\begin{lemma}[Basis characterization]\label{lem:basis-characterization}
    Let \(\mathcal B\) be a basis for a topology \(\tau\) on \(X\). Then, 
for every \(U\subseteq X\), $U\in\tau
\iff
\forall x\in U,\ \exists B\in\mathcal B,
\text{ such that }x\in B\subseteq U$.
\end{lemma}

The \textit{Cantor topology} on \(\Sigma^\omega\) is generated by the basic open sets $s\cdot\Sigma^\omega=\{\,t\in\Sigma^\omega\mid s\preceq t\,\}$, with $s\in\Sigma^*$. Thus, $\tau_{\mathrm{Cantor}}
=
\left\langle
\{\,s\cdot\Sigma^\omega\mid s\in\Sigma^*\,\}
\right\rangle$.

Finally, we recall the \textit{Alexandroff topology} induced by a preorder. Let
\(\sqsubseteq\) be a preorder on a set \(X\). A set \(U \subseteq X\) is
upward-closed with respect to \(\sqsubseteq\) if, whenever \(x\in U\) and
\(x\sqsubseteq y\), then \(y\in U\). The Alexandroff topology induced by
\(\sqsubseteq\) is $\tau(\sqsubseteq)
=
\{\,U\subseteq X\mid U \text{ is upward-closed with respect to } \sqsubseteq\,\}$.

\section{Trace-Based Topologies}
\label{sec:trace-based topologies}

The topological approach proceeds in two stages. We first study the mathematical structure induced by each observation mechanism, asking whether
the corresponding observation sets form a basis and how the resulting topologies compare. We then use these topological constructions to characterize the verification and monitorability induced by them.

We first consider observations based only on traces. This provides a baseline in which processes are compared through the traces they can
generate, capturing the linear behaviour of executions rather than their branching structure. The trace preorder reflects this view by ordering
processes according to trace inclusion, and the associated Alexandroff topology provides the corresponding topological perspective.

\subsection{Full Trace Inclusion and the Alexandroff Topology}

\begin{definition}[(Full) Trace preorder]
Define a relation \(\sqsubseteq_{T}\) on \(\Proc\) by $p\sqsubseteq_{T} q
\iff
\Trc(p)\subseteq \Trc(q)$.
\end{definition}

The relation \(\sqsubseteq_{T}\) is a preorder: it is reflexive and transitive because set inclusion is reflexive and transitive. Although \(\sqsubseteq_{T}\) is defined on processes, it compares them only through their traces. This motivates the notation \q{\(T\)}, which emphasizes the trace-based nature of the comparison.

As we saw in Section \ref{preliminaries}, a preorder over $\Proc$ can be used to define its associated Alexandroff topology, where the open sets are the upward-closed sets. For the trace preorder, this means that if a process belongs to an open set, then any process whose trace set contains the traces of that process also belongs to the same open set.

\begin{definition}[Alexandroff topology induced by \(\sqsubseteq_{T}\)]
The \textit{Alexandroff topology} induced by the trace preorder \(\sqsubseteq_{T}\) is $\tau(\sqsubseteq_{T})
=
\{\,U\subseteq \Proc \mid U \text{ is upward-closed with respect to } \sqsubseteq_{T}\,\}$.
Equivalently, \(U\in\tau(\sqsubseteq_{T})\), \text{iff} $
\forall p\in U\;\forall q\in \Proc\;(p\sqsubseteq_{T} q \Rightarrow q\in U)$.
\end{definition}

\subsection{Trace Processes and the Collapse of Full Trace Inclusion}

In the following, we turn to a special case in which a process can execute only one infinite trace. Such processes may still have several finite traces, but their infinite behaviour is uniquely determined.

\begin{definition}[Trace-process]
A state \(p\in \Proc\) is a \emph{trace-process} if $\bigl|\Trc(p)\cap \Sigma^\omega\bigr|=1$.
We write $\Trc(p)\cap \Sigma^\omega=\{t_p\}$, 
where \(t_p\in\Sigma^\omega\) is the unique infinite trace of \(p\).
\end{definition}

The following lemma shows that trace inclusion between trace processes forces equality of the corresponding infinite traces.

\begin{lemma}
\label{lem:unique-trace-equality-contrib}
Let \(p,q\in \Proc\) be trace-processes. If $\Trc(p)\subseteq \Trc(q)$, then
$t_p=t_q$.
\end{lemma}

Lemma \ref{lem:unique-trace-equality-contrib} shows that two trace-processes comparing with trace inclusion must have the same infinite trace. However, they may still differ in their finite traces. To avoid this difference, we introduce a stronger linearity condition. 

\begin{definition}[Linear trace-process]
\label{linear trace-process}
A trace-process \(p\) is called \emph{linear} if $\Trc(p)=\Pref(t_p)\cup\{t_p\}$, where \(\Pref(t_p)\subseteq \Sigma^*\) is the set of finite prefixes of \(t_p\).
\end{definition}

Thus, a linear trace-process has exactly one infinite trace, and its finite
traces are precisely the finite prefixes of that trace. Consequently, if \(p\)
and \(q\) are linear trace-processes such that
\(\Trc(p)\subseteq\Trc(q)\), then \(\Trc(p)=\Trc(q)\).

The following example illustrates that linearity restricts the trace behaviour
of a process, but does not prevent the process from having countably many
branches with different finite behaviours. Here, \(a^0.0=0\) and
\(a^{n+1}.0=a.(a^n.0)\), for every \(n\geq 0\); and \(a^\omega\) is the process with only one transition: $a^\omega \xrightarrow{a} a^\omega$.

\begin{example}
\label{ex:linear-countably-branching}
Fix \(a\in\Sigma\), and let \(p=a^\omega\). Define $q=p+\sum_{n\geq 0} a^n.0$ and $r=\sum_{n\geq 0} a^n.0$.
Since $\Trc(p)=\Pref(a^\omega)\cup\{a^\omega\}$, and every finite trace contributed by the additional summands of \(q\) is
already a prefix of \(a^\omega\), we have
$\Trc(q)=\Trc(p)$.
Hence, both \(p\) and \(q\) are linear trace-processes. Nevertheless, \(q\)
has countably many \(a\)-successors with different finite behaviours, meaning that after
taking an \(a\)-transition, they can perform different numbers of further
\(a\)-actions before stopping. In contrast, \(r\) has no infinite trace and is
therefore not a trace-process.
\end{example}

\begin{lemma}
\label{lem:order-collapses-contrib}
Let \(p,q\in \Proc\) be linear trace-processes. Then, $p\sqsubseteq_{T} q
\Longleftrightarrow
\Trc(p)=\Trc(q)$. 
In particular, if \(p\sqsubseteq_{T} q\), then \(p\) and \(q\) have the same unique infinite trace.
\end{lemma}

Thus, on linear trace-processes, full trace inclusion collapses to equality of trace sets. It does not necessarily collapse to equality of processes, since distinct linear trace-processes may have the same trace set, as illustrated by
Example~\ref{ex:linear-countably-branching}. We now specialize to the trace transition system, where each process is itself an infinite trace. In that setting, equality of trace sets will imply equality of processes.

\subsection{\texorpdfstring{Finite Trace Observations and the Topology \(\tau_O\)} {Finite Trace Observations and the Topology tau-O}}
\label{sec:finite trace observations}

We instantiate the previous trace-based setting to 
the process space of infinite traces and make the connection between finite trace observations and the standard topology on infinite words explicit.

\paragraph{\texorpdfstring{The Trace Transition System $LTS_{\mathrm{tr}}$.}{The Trace Transition System LTS-tr}}

We implement the previous constructions in the trace transition system $LTS_{\mathrm{tr}}=\langle \Sigma^\omega,\Sigma,\to\rangle$, defined by $t_1\xrightarrow{a} t_2
\Longleftrightarrow
t_1=a \cdot t_2$.
Here, the process space is \(\Sigma^\omega\), and each transition removes the first letter of an infinite word. Thus, for every $t\in\Sigma^\omega$, $\Trc(t)=\Pref(t)\cup\{t\}$.
Indeed, a finite word $s\in\Sigma^*$ is executable from $t$ exactly when $s\preceq t$, since each transition removes the next letter of $t$. Moreover, the only infinite trace executable from $t$ is $t$ itself. Hence, every process in $LTS_{\mathrm{tr}}$ is a linear trace-process.

\paragraph{\texorpdfstring{Finite Positive Trace Observations on \(\Sigma^\omega\).}{Finite Positive Trace Observations on Sigma omega}}

Having identified the executable traces in \(LTS_{\mathrm{tr}}\), we now define the finite observations generated by these traces. For a finite word \(s\), the \textit{observation} \(O(s)\) collects all infinite traces that have \(s\) as an executable finite trace, \emph{i.e.}, that have \(s\) as a prefix.

\begin{definition}[Finite trace observation]
For every \(s\in\Sigma^*\), define $O(s)=\{\,t\in\Sigma^\omega \mid s\in\Trc(t)\,\}$.
\end{definition}

\begin{remark}
\label{rem:item4-trace-LTS-contrib}
By the definitions of \(LTS_{\mathrm{tr}}\) and $O(-)$, for every \(s\in\Sigma^*\), $O(s)=s\cdot\Sigma^\omega$. Hence, the finite trace observations are precisely the basic open sets of
the Cantor topology, and the topology they generate is therefore \(\tau_{\mathrm{Cantor}}\), \emph{i.e.},
$\tau_O
=
\left\langle\mathcal B_O\right\rangle
=
\tau_{\mathrm{Cantor}}$.
\end{remark}

\begin{proposition}[Finite trace observations induce the Cantor topology]
\label{prop:BS-cantor-trace-LTS}
$\mathcal B_O
=
\{\,O(s)\mid s\in\Sigma^*\,\}$ 
is the standard basis of the Cantor topology on
\(\Sigma^\omega\). 
\end{proposition}

This gives the first key conclusion of this section. When observations are finite traces, the resulting topology on \(\Sigma^\omega\) is not discrete, but coincides with the classical Cantor topology.

\subsubsection{Finite-Depth Trace Inclusion}

We now introduce finite-depth trace inclusion for an arbitrary labelled transition system. Instead of comparing full trace sets, this relation compares
only traces of a fixed finite length \(k\). We later specialize this construction to the trace transition system \(LTS_{\mathrm{tr}}\).

\begin{definition}[Finite-depth trace inclusion]
\label{def: Finite-depth trace inclusion}
Let \(L=\langle\Proc,\Sigma,\to\rangle\) be an LTS.
For every \(k\in\mathbb N\) and every \(p\in\Proc\), define $\Trc_k(p)=\Trc(p)\cap\Sigma^k$. For \(p,q\in\Proc\), define $
p\sqsubseteq_T^k q
\Longleftrightarrow
\Trc_k(p)\subseteq\Trc_k(q)$. 
Let \(\tau(\sqsubseteq_T^k)\) denote the Alexandroff topology induced by
\(\sqsubseteq_T^k\). The finite-depth trace topology is defined as $\tau(\sqsubseteq_T^{\mathrm{fin}})
=
\left\langle
\bigcup_{k\geq 0}\tau(\sqsubseteq_T^k)
\right\rangle$. 
\end{definition}

We now specialize this construction to the trace transition system
\(LTS_{\mathrm{tr}}\). Since every finite trace of
\(t\in\Sigma^\omega\) is a prefix of \(t\), the \(k\)-trace information of
\(t\) consists of exactly its prefix of length \(k\).

\begin{lemma}
\label{lem:Trck-singleton-trace-LTS-contrib}
For every \(t\in\Sigma^\omega\) and every \(k\in\mathbb N\), $\Trc_k(t)=\{\,t\!\upharpoonright k\,\}$.
\end{lemma}

It follows that finite-depth trace inclusion is equality of prefixes of length \(k\).

\begin{lemma}
\label{lem:preorder-prefix-trace-LTS-contrib}
For all \(t,u\in\Sigma^\omega\) and all \(k\in\mathbb N\), $t\sqsubseteq_{T}^k u
\Longleftrightarrow
t\!\upharpoonright k=u\!\upharpoonright k$.
\end{lemma}

Hence, the open sets determined by \(\sqsubseteq_{T}^k\) are the sets of infinite words sharing a prefix of length \(k\).

\begin{corollary}
\label{cor:Tk-prefix-sets-contrib}
For each \(k\in\mathbb N\) and each \(s\in\Sigma^k\), the set $
\{\,t\in\Sigma^\omega \mid t\!\upharpoonright k=s\,\}$ is exactly $s\cdot\Sigma^\omega$. Thus, the sets determined by \(\sqsubseteq_{T}^k\) are exactly the \(k\)-length Cantor basic open sets.
\end{corollary}

Combining all finite depths, therefore, yields the same topology obtained from finite trace observations.

\begin{theorem}
\label{thm:item2-trace-LTS-contrib}
For the trace transition system \(LTS_{\mathrm{tr}}\), the topology $\tau(\sqsubseteq_{T}^{\mathrm{fin}})$
coincides with the Cantor topology on \(\Sigma^\omega\).
\end{theorem}

Thus, finite trace observations and finite-depth trace inclusion lead to the same topology on the space of infinite traces \(\Sigma^\omega\). This shows that, although the two constructions are defined differently, both capture the same finite-prefix information.

\subsubsection{Finite versus Full Trace Topologies}

We now compare the finite constructions with the topology induced by full trace inclusion, and
show that finite and full trace information do not have the same topological effects.

Since every process in \(LTS_{\mathrm{tr}}\) is a linear trace-process, Lemma~\ref{lem:order-collapses-contrib} yields $\Trc(t)=\Trc(u)$, whenever \(t\sqsubseteq_T u\). Each process \(t\in\Sigma^\omega\) has \(t\) as its unique infinite trace. Hence, equality of trace sets implies \(t=u\). Therefore, $t\sqsubseteq_T u \iff t=u$, for every \(t,u\in\Sigma^\omega\). Thus, \(\sqsubseteq_T\) is equality on \(\Sigma^\omega\), and every subset of \(\Sigma^\omega\) is upward-closed. Consequently, $\tau(\sqsubseteq_T)=\mathcal P(\Sigma^\omega)$, so \(\tau(\sqsubseteq_T)\) is the discrete topology on \(\Sigma^\omega\).

\begin{corollary}
\label{cor:trace-LTS-summary}
For the trace transition system \(LTS_{\mathrm{tr}}\), we have $\left\langle \{O(s)\mid s\in\Sigma^*\}\right\rangle
=
\tau(\sqsubseteq_{T}^{\mathrm{fin}})
=
\tau_{\mathrm{Cantor}}$. Moreover, $\tau(\sqsubseteq_{T})=\mathcal P(\Sigma^\omega)$. Hence, if \(|\Sigma|\ge 2\), then $\left\langle \{O(s)\mid s\in\Sigma^*\}\right\rangle
=
\tau(\sqsubseteq_{T}^{\mathrm{fin}})
\subsetneq
\tau(\sqsubseteq_{T})$.
\end{corollary}

\begin{remark}
When $|\Sigma|=1$, say $\Sigma=\{a\}$, we have $\Sigma^\omega=\{a^\omega\}$. Hence, $\tau_{\mathrm{Cantor}}=\mathcal{P}(\Sigma^\omega)$ is discrete, and therefore
$
\left\langle \{O(s)\mid s\in\Sigma^*\}\right\rangle
=
\tau(\sqsubseteq_{T}^{\mathrm{fin}})
=
\tau(\sqsubseteq_{T})
=
\mathcal{P}(\Sigma^\omega)$. Thus, the strict inclusion in Corollary \ref{cor:trace-LTS-summary} occurs precisely when $|\Sigma|\geq 2$.
\end{remark}

\section{Process-Based Topologies}
\label{sec:process-topologies}

In Section \ref{sec:finite trace observations}, we studied the special case in which the process space is \(\Sigma^\omega\). We now move to a general process space, where a process may have several possible traces and may also admit a branching structure. This change of setting is important because trace information alone records which executions are possible, but not how these executions are organized.

\subsection{\texorpdfstring{Finite Positive Trace Observations on $\Proc$}{Finite Positive Trace Observations on Proc}}

Let \(\Proc\) now be an arbitrary process space. We again consider finite
positive trace observations of the form $
O(s)=\{\,p\in\Proc\mid s\in\Trc(p)\,\}$, $s\in\Sigma^*$. The meaning of \(O(s)\) is the same as in \ref{sec:finite trace observations}. It collects all processes that
can execute the finite trace \(s\). However, its role is different from the
trace-space case. When the process space is \(\Sigma^\omega\), each process is
itself a single infinite trace, and \(O(s)\) is simply the set of all infinite
words with prefix \(s\). In a general process space, a process may have many
possible traces and may also branch internally. Thus, \(O(s)\) records only that
the trace \(s\) is executable. It does not record the branching structure that
leads to this trace.

Since a process may satisfy several finite trace observations at the same time,
we must also consider finite intersections of such observations. For example, $O(s_1)\cap\cdots\cap O(s_n)$
contains exactly those processes that can execute all of the finite traces
\(s_1,\ldots,s_n\). In general, such an intersection cannot be reduced to a
single observation \(O(s)\). For this reason, the family
\(\{O(s)\mid s\in\Sigma^*\}\) is treated as a subbasis, and its finite
intersections form the corresponding basis.

For every finite set \(D\in\mathcal P_f(\Sigma^*)\), define $O(D)=\bigcap_{s\in D}O(s)$.
Thus, \(O(D)\) is the set of processes that can execute every trace in \(D\).
We define  $O(\varnothing)=\bigcap_{s\in\varnothing}O(s)=\Proc$, since the empty observation places no requirements on a process.

\begin{definition}[Topology $\tau_O$ on processes]
We define $\tau_O=\left\langle \{O(s)\mid s\in\Sigma^*\}\right\rangle$.
\end{definition}

The finite intersections of these subbasic observations form the corresponding basis.

\begin{proposition}
\label{prop:BS-process-basis-contrib}
The family $\mathcal B_O=
\left\{
\bigcap_{i=1}^n O(s_i)
\;\middle|\;
n\ge 1,\ s_1,\dots,s_n\in\Sigma^*
\right\}$
is a basis for \(\tau_O\). In particular, \(\tau_O\) is a topology on \(\Proc\).
\end{proposition}

The topology \(\tau_O\) captures properties that can be verified by finitely many positive trace observations. We now relate this topology to the finite-depth trace inclusion topology introduced in Definition \ref{def: Finite-depth trace inclusion}. The following theorem shows that these two constructions are equivalent on arbitrary LTSs.

\begin{theorem}
\label{equality of topologies on processes}
For every LTS, $\tau_O=\tau(\sqsubseteq_{T}^{\mathrm{fin}})$.
\end{theorem}

Thus, \(\tau_O\) captures finite trace information on processes. To move beyond
linear traces and capture finite branching structure, we next introduce
simulation-based observations.

\subsection{\texorpdfstring{The Simulation Observation Topology \(\tau_{\mathrm{sim}}\)} {The Simulation Observation Topology tau-sim}}

The goal in this subsection is to define a topology whose basic observations are finite process behaviours rather than finite traces. 

Intuitively, the observation types considered in this paper differ in the amount of behavioural information they retain about a process. A finite trace records only a sequence of actions that a process can perform and therefore captures linear behavioural information. A finite simulation observation additionally captures branching structure by describing possible continuations that must be matched by another process. Complete simulation and ready simulation refine this information further by imposing additional requirements on the availability of actions. Finally, finite-depth bisimulation observations compare the behaviour of two processes in both directions up to a bounded depth.

For the simulation topology, we use finite loop-free processes as observations. A process satisfies such an observation when it can simulate the finite branching structure represented by that observation.

Recalling the definitions of finite loop-free processes and the simulation preorder introduced in Section \ref{preliminaries}, we now define simulation observations.

\begin{definition}[Basic simulation observation]
For \(q\in \mathrm{Proc_F}\), define $O(q)=\{\,p\in \Proc\mid q \sqsubseteq_S p\,\}$, 
where \(q \sqsubseteq_S p\) means that \(p\) simulates the finite loop-free process \(q\).
\end{definition}

The set \(O(q)\) contains all processes that satisfy the finite simulation observation \(q\). Unlike a trace observation ($O(s)$), which only checks whether a process can execute one linear trace, a simulation observation $(O(q))$ can require a process to match a finite branching behaviour.
Hence, simulation observations can express finite branching information. 

A useful property of simulation observations is that finite intersections can still be expressed by a single observation. The operator \(+\) provides this representation, since a process simulates \(q_1+q_2\) exactly when it simulates both \(q_1\) and \(q_2\).

\begin{lemma}
\label{lem:plus-sim-contrib}
For all \(q_1,q_2\in \mathrm{Proc_F}\) and all \(p\in \Proc\), 
$q_1+q_2\sqsubseteq_S  p
\iff
q_1 \sqsubseteq_S  p \text{ and } q_2 \sqsubseteq_S  p$.
\end{lemma}

This immediately yields the basis structure of the new topology.

\begin{proposition}
\label{prop:Bsim-basis-contrib}
The family $\mathcal B_{\mathrm{sim}}=\{\,O(q)\mid q\in \mathrm{Proc_F}\,\}$ is a basis for a topology on \(\Proc\).
\end{proposition}

\begin{definition}[Simulation topology]
Let $\tau_{\mathrm{sim}}$ be the topology on \(\Proc\) generated by the basis \(\mathcal B_{\mathrm{sim}}\).
\end{definition}

This topology is process-based because its basic sets no longer examine whether specific traces are present, but whether a process can match a given finite branching behaviour.

\subsection{\texorpdfstring{$\tau_O$ versus $\tau_{\mathrm{sim}}$}{tau O versus tau sim}}

We now compare \(\tau_O\) with \(\tau_{\mathrm{sim}}\). Both topologies are generated by finite observations, but they observe different kinds of information. The topology \(\tau_O\) is based on finite traces, while \(\tau_{\mathrm{sim}}\) is based on finite loop-free processes. Thus, \(\tau_{\mathrm{sim}}\) is in general strictly more expressive than \(\tau_O\), because it can distinguish finite branching patterns that finite
trace observations cannot observe.

We first show that every trace observation can be encoded by a suitable finite loop-free process, which implies that every \(\tau_O\)-open set is also \(\tau_{\mathrm{sim}}\)-open.

\begin{theorem}
\label{thm:tauS-in-tauSim-contrib}
Let \(L=\langle\Proc,\Sigma,\to\rangle\) be an LTS. Then, $\tau_O\subseteq\tau_{\mathrm{sim}}$.
\end{theorem}

\begin{proof}[Proof idea]
Every finite trace observation can be represented as a simulation observation.
Indeed, if \(s=a_1\cdots a_n\), then \(s\) can be encoded by the finite linear
process $q_s=a_1.\cdots a_n.0$.
A process can execute the trace \(s\) exactly when it simulates \(q_s\). Hence,
each subbasic open set \(O(s)\) of \(\tau_O\) is also open in
\(\tau_{\mathrm{sim}}\).
\end{proof}

The inclusion of Theorem \ref{thm:tauS-in-tauSim-contrib} says that simulation observations are at least as expressive as finite trace observations. The next example shows that they are strictly more expressive in general.

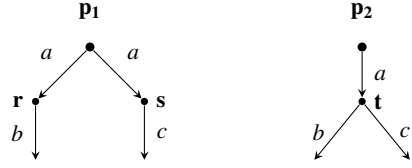
\begin{figure}[t]
\centering
\begin{tikzpicture}[scale=0.8, every node/.style={scale=0.8}, >=stealth]

\node at (0,2.6) {$\mathbf{p_1}$};

\node[circle,fill,inner sep=1.5pt] (p1) at (0,2) {};
\node[circle,fill,inner sep=1.2pt,label=left:{$\mathbf{r}$}] (p1l) at (-0.9,1.1) {};
\node[circle,fill,inner sep=1.2pt,label=right:{$\mathbf{s}$}] (p1r) at (0.9,1.1) {};
\node (p1lb) at (-0.9,0) {};
\node (p1rc) at (0.9,0) {};

\draw[->] (p1) -- (p1l) node[pos=0.45, above left=2pt] {$a$};
\draw[->] (p1) -- (p1r) node[pos=0.45, above right=2pt] {$a$};
\draw[->] (p1l) -- (p1lb) node[midway, left=2pt] {$b$};
\draw[->] (p1r) -- (p1rc) node[midway, right=2pt] {$c$};

\node at (4.5,2.6) {$\mathbf{p_2}$};

\node[circle,fill,inner sep=1.5pt] (p2) at (4.5,2) {};
\node[circle,fill,inner sep=1.2pt,label=right:{$\mathbf{t}$}] (p2m) at (4.5,1.1) {};
\node (p2l) at (3.6,0) {};
\node (p2r) at (5.4,0) {};

\draw[->] (p2) -- (p2m) node[midway, right=2pt] {$a$};
\draw[->] (p2m) -- (p2l) node[midway, left=2pt] {$b$};
\draw[->] (p2m) -- (p2r) node[midway, right=2pt] {$c$};

\end{tikzpicture}
\caption{Two trace-equivalent processes with different branching structures.}
\label{fig:p1-p2-processes}
\end{figure}

\begin{example}[Strictness of the inclusion]
\label{exp:tauS-strictly-in-tauSim-contrib}
Consider the processes \(p_1,p_2\in\Proc\), illustrated in
Figure~\ref{fig:p1-p2-processes}. The process \(p_1\) has exactly two
\(a\)-successors \(r\) and \(s\), where \(r\) has a \(b\)-transition but no
\(c\)-transition, and \(s\) has a \(c\)-transition but no \(b\)-transition. The
process \(p_2\) has an \(a\)-successor \(t\) with both a \(b\)-transition and a
\(c\)-transition. Thus, \(p_1\) and \(p_2\) have the same finite traces, that is $\Trc(p_1)=\Trc(p_2)$. In particular, both processes can execute the traces \(ab\) and \(ac\).

This means that \(\tau_O\) cannot distinguish \(p_1\) and \(p_2\), because \(\tau_O\) is based only on finite trace observations. However, \(\tau_{\mathrm{sim}}\) can observe the branching structure of the traces. To see this, consider the finite loop-free process $q=a.(b.0+c.0)$. Here, $q$ requires a process to execute an \(a\)-transition to a state which can then execute both \(b\) and \(c\).
The process \(p_2\) satisfies this observation, while \(p_1\) does not. Therefore, $q \sqsubseteq_S p_2$, while $q\not\sqsubseteq_S p_1$. Consequently, \(\tau_{\mathrm{sim}}\) can separate \(p_1\) and \(p_2\), while \(\tau_O\)
cannot. This shows that simulation observations contain information that is not captured by finite trace observations. 
\end{example}

\begin{corollary}[Strictness of the inclusion]
\label{cor:tauS-in-tauSim-strict}
Let \(L=\langle\Proc,\Sigma,\to\rangle\) be an LTS. Assume there exist processes \(p_1,p_2\in\Proc\) and a finite loop-free process \(q\in\mathrm{Proc}_F\), such that
$\Trc(p_1)=\Trc(p_2)$, $q\sqsubseteq_S p_2$, and $q\not\sqsubseteq_S p_1$. In particular, if \(\mathrm{Proc}_F\subseteq\Proc\), then $\tau_O\subsetneq\tau_{\mathrm{sim}}$.
\end{corollary}

\section{Verification Induced by Observation Topologies}
\label{sec:verification}

Sections \ref{sec:trace-based topologies} and \ref{sec:process-topologies} established the topological structures induced by the different finite behavioural observations considered in this paper. These constructions provide the observational foundation for the verification results that follow. Once an observation family forms a basis, its open sets characterize exactly the
properties whose satisfaction can be established from finite observations. We now make this connection explicit by developing a general topological verification approach and then instantiating it for the trace and simulation topologies.

Let \(\Obs\) be a set of finite observations, and let $O:\Obs\to\mathcal P(\Proc)$ assign to each observation \(o\in \Obs\) the set \(O(o)\) of processes
satisfying that observation. 
Thus, $O(o) = \{ p \in \Proc \mid p \text{ satisfies } o\,\}$.

Assume that $\mathcal B_{G}=\{\,O(o)\mid o\in\Obs\,\}$ 
is a basis for a general topology \(\tau_{G}\) on \(\Proc\).

\subsection{Topological Verification} 
\label{topol-verif}

Throughout this subsection, verification means \textit{satisfaction} verification, that is, a verifier is used to confirm that a process belongs to a property \(L\).

We now define verification with respect to the basis \(\mathcal B_G\). Here, a verifier receives a finite observation as input and determines whether that observation provides sufficient evidence to confirm membership
in the property being verified. We represent the verifier by a characteristic function, where \(1\) denotes a positive verification
verdict and \(0\) denotes that no conclusive positive verdict is reached.

\begin{definition}[Topological verifier]
\label{Topological verifier}
Let \(L\subseteq\Proc\). A $\mathcal B_{G}$-\emph{verifier} for \(L\) is a function $v: \Obs\to\{0,1\}$, where \(v(o)=1\) indicates a positive verification verdict and
\(v(o)=0\) indicates that no conclusive positive verdict is reached, such that, for every \(o,o'\in \Obs\):
\begin{enumerate}
\item[i.] \emph{(Soundness.)} If \(v(o)=1\), then \(O(o)\subseteq L\).
\item[ii.] \textit{(Monotonicity.)} If \(v(o)=1\) and \(O(o')\subseteq O(o)\), then
\(v(o')=1\).
\end{enumerate}
\end{definition}

\begin{definition}[Topological verifiability]
A property \(L\subseteq\Proc\) is $\mathcal B_{G}$-\emph{verifiable} if there exists a $\mathcal B_{G}$-verifier \(v\) for \(L\), such that, for every \(p\in L\), there exists \(o\in \Obs\) with $p\in O(o)$ and $v(o)=1$. 
\end{definition}

\begin{theorem}
\label{thm:general-topological-verification}
A property \(L\subseteq\Proc\) is $\mathcal B_{G}$-verifiable if and only if
$L\in \tau_{G}$.
\end{theorem}

\begin{proof}[Proof idea]
The equivalence follows from the basis characterization of open sets. If
\(L\) is open, then for every \(p\in L\), there exists a basic neighbourhood
\(O(o)\), such that \(p\in O(o)\subseteq L\). The verifier \(v\) can therefore be defined to return \(1\) on precisely
those observations \(o\) for which \(O(o)\subseteq L\). Conversely, if a
verifier returns \(1\) on an observation \(o\), soundness gives \(O(o)\subseteq L\). Since every \(p\in L\) belongs to some accepted observation, \(L\) is a union of basic open sets and therefore
\(L\in\tau_G\).
\end{proof}

\subsection{Monitorability Induced by \texorpdfstring{$\tau_O$}{tau-O} and \texorpdfstring{$\tau_{\mathrm{sim}}$}{tau-sim}}
\label{sec:monitorability}

The general verification result from Subsection~\ref{topol-verif} applies directly to the observation topologies developed in the previous sections. In particular, it allows us to interpret the open sets of \(\tau_O\) and \(\tau_{\mathrm{sim}}\) in terms of concrete monitorability
notions. For \(\tau_O\), the observations are finite sets of finite traces, whereas for \(\tau_{\mathrm{sim}}\), the observations are finite loop-free processes. Thus, the two topologies give rise to two different notions of monitorability, depending on the behavioural information available to the monitor.

As with verification, throughout this subsection, monitorability refers to \textit{satisfaction} monitorability, where finite observations are used to confirm that a process belongs to a property \(L\). We do not separately consider violation monitorability, since detecting violations of \(L\) can be viewed as satisfaction monitorability of its complement
\(\Proc\setminus L\).

We also note that monitorability is usually defined with respect to observations that consist of a single finite trace \cite{Aceto2021,inbook}. The following definitions, however, assume that monitors can verify a system by observing multiple executions (\emph{i.e.}, traces), which is a reasonable assumption from the point of view of Runtime Verification, and it fits well with our topological point of view.

\subsubsection{Multi-Trace Monitorability}
\label{sec:multi-tr mon}

For \(\tau_O\), an observation is a finite set of finite traces. Let \(\mathcal P_f(\Sigma^*)\) denote the set of all finite subsets of \(\Sigma^*\). For \(D=\{s_1,\dots,s_n\}\in\mathcal P_f(\Sigma^*)\), define $O(D)=\bigcap_{s\in D}O(s)$. Thus, \(O(D)\) is the set of processes that can execute every finite trace
in \(D\). We let $\Obs_O=\mathcal P_f(\Sigma^*)$.

\begin{definition}[Multi-trace monitor]
\label{def:multi-trace-monitor-ms}
Let \(L\subseteq\Proc\). A \emph{multi-trace monitor} for \(L\) is a characteristic function $m:\mathcal P_f(\Sigma^*)\to\{0,1\}$, where \(m(D)=1\) indicates that the finite trace set \(D\) is accepted, and \(m(D)=0\) indicates that \(D\) does not provide a conclusive acceptance verdict. For every \(C,D\in\mathcal P_f(\Sigma^*)\), it satisfies:
\begin{enumerate}
\item[i.] \emph{(Soundness.)} If \(m(D)=1\), then
$\bigcap_{s\in D} O(s)\subseteq L$.
\item[ii.] \textit{(Monotonicity.)} If \(m(C)=1\) and \(C\subseteq D\), then $m(D)=1$. 
\end{enumerate}
\end{definition}

\paragraph{Topological Formulation of Multi-Trace Monitorability.}

We first characterize monitorability in terms of the basic open sets induced by finite trace observations. A finite set of traces \(D\) provides a conclusive acceptance verdict for a property \(L\) when the corresponding observation set \(O(D)\) is contained in \(L\). Thus, a property is multi-trace monitorable when every process satisfying the property has a finite trace observation that is accepted by the monitor.

\begin{definition}[Multi-trace monitorability, Topologically]
\label{def: multi-trace mon sense 1}
A property \(L\subseteq \Proc\) is \emph{multi-trace monitorable} if there exists a multi-trace monitor \(m\) for \(L\), such that for every \(p\in L\), there exists a finite set $D\subseteq \Trc(p)\cap\Sigma^*$, with $m(D)=1$.
\end{definition}

Theorem \ref{thm:general-topological-verification} immediately yields the corresponding characterization for \(\tau_O\). In this case, verifiability coincides with multi-trace monitorability.

\begin{corollary}
\label{Equivalence L top with L in tau s}
A property \(L\subseteq \Proc\) is multi-trace monitorable if and only if $
L\in\tau_O$.
\end{corollary}

\paragraph{Operational Formulation of Multi-Trace Monitorability.}

The topological formulation characterizes monitorability through the observation sets that form the basic opens of \(\tau_O\). We now give an equivalent operational formulation based on the acceptance of individual processes.

\begin{definition}[Acceptance of a process]
\label{def:process-trace-acceptance}
Let $m:\mathcal P_f(\Sigma^*)\to\{0,1\}$ be a multi-trace monitor. For \(p\in\Proc\), define $\mathrm{acc}_m(p)
\Longleftrightarrow
\exists D\subseteq\Trc(p)\cap\Sigma^*,
\text{ such that }m(D)=1$.
\end{definition}

\begin{definition}[Multi-trace monitorability, Operationally]
\label{def:multi-trace-satisfaction-monitorability}
A property \(L\subseteq\Proc\) is \emph{multi-trace monitorable} if there exists a multi-trace monitor $m:\mathcal P_f(\Sigma^*)\to\{0,1\}$, and such that for every \(p\in L\), $\mathrm{acc}_m(p)$.
\end{definition}

The two formulations capture the same notion of monitorability from different perspectives. The topological formulation expresses monitorability through the basic open sets induced by finite trace
observations, whereas the operational formulation expresses it through the acceptance of individual processes. The following theorem establishes the equivalence between these formulations and the corresponding topological characterization.

\begin{theorem}
\label{theorem:topological-monitoring-system-equivalence}
Let \(L\subseteq \Proc\). The following are equivalent: 
\begin{enumerate}[label=\textup{(\roman*)}
]
    \item[i.] \(L\in\tau_O\).
    \item[ii.] \(L\) is multi-trace monitorable according to
    Definition~\ref{def: multi-trace mon sense 1}.
    \item[iii.] \(L\) is multi-trace monitorable according to
    Definition~\ref{def:multi-trace-satisfaction-monitorability}.
\end{enumerate}
\end{theorem}

\subsubsection{Simulation Monitorability}

For \(\tau_{\mathrm{sim}}\), an observation is a finite loop-free process.
Let $\Obs_{\mathrm{sim}}=\mathrm{Proc_F}$. For \(q\in\Obs_{\mathrm{sim}}\), define $O(q)=\{\,p\in\Proc\mid q \sqsubseteq_S p\,\}$.
Since the basic open sets of \(\tau_{\mathrm{sim}}\) are the sets \(O(q)\), a simulation monitor receives as an input a finite loop-free process \(q\). The monitor determines whether the observation provides sufficient evidence
to accept the property \(L\), using the same characteristic-function representation as in \ref{sec:multi-tr mon}.

\begin{definition}[Simulation monitor (Topological formulation)]
\label{simulation monitor}
Let \(L\subseteq \Proc\). A \emph{simulation monitor} for \(L\) is a characteristic function $m:\mathrm{Proc}_F\to\{0,1\}$, 
where \(m(q)=1\) indicates acceptance of the finite simulation observation
\(q\), and \(m(q)=0\) indicates that no conclusive acceptance verdict is
reached. It satisfies:
\begin{itemize}
\item[i.] \emph{(Soundness.)} If \(m(q)=1\), then $O(q)\subseteq L$, where $O(q)=\{\,p\in \mathrm{Proc}\mid q \sqsubseteq_S p\,\}$.
\item[ii.] \textit{(Monotonicity.)} If \(m(q)=1\) and $O(r)\subseteq O(q)$, for $r,q \in \mathrm{Proc_F}$, then $m(r)=1$.
\end{itemize}
\end{definition}

\begin{definition}[Simulation monitorability]
\label{def: sim-mon}
A property \(L\subseteq \Proc\) is \emph{simulation-monitorable} if there exists a simulation monitor \(m\) for \(L\), such that for every \(p\in L\), there exists \(q\in \mathrm{Proc_F}\) with $p\in O(q)$ and $m(q)=1$. Equivalently, since \(p\in O(q)\) iff \(q\sqsubseteq_S p\), for every \(p\in L\), there exists \(q\in \mathrm{Proc_F}\), such that $q \sqsubseteq_S p$ and $m(q)=1$.
\end{definition}

Definition \ref{def: sim-mon} corresponds directly to the characterization
of open sets in \(\tau_{\mathrm{sim}}\). A process \(p\in L\) must have a finite simulation
observation \(q\), such that \(p\in O(q)\) and all processes satisfying the same
observation belong to \(L\).
Theorem \ref{thm:general-topological-verification} also gives the corresponding characterization for the simulation topology. Here, the finite observations are finite loop-free processes, and verifiability coincides with simulation monitorability.

\begin{corollary}
\label{simul mon in tau sim}
A property \(L\subseteq \mathrm{Proc}\) is simulation-monitorable if and only if $L\in\tau_{\mathrm{sim}}$. 
\end{corollary}

\subsection{Beyond Simulation Observations}
\label{sec:beyond simulation}

In Corollary \ref{simul mon in tau sim}, we showed that simulation monitorability corresponds exactly to openness in \(\tau_{\mathrm{sim}}\). A natural next step is to ask whether the same construction works when simulation is replaced by stronger behavioural relations. The key question is whether the corresponding finite observations still provide a basis for
a topology. This is motivated by the standard hierarchy of preorders \cite{GLABBEEK20013}, where ready simulation and complete simulation appear as strict refinements of ordinary simulation.

Throughout this subsection, we consider an LTS \(L=\langle\Proc,\Sigma,\to\rangle\) with
\(\mathrm{Proc_F}\subseteq\Proc\), where \(\mathrm{Proc}_F\) denotes the finite loop-free processes used as observations. Thus, \(\Proc\) is the process space, while \(\mathrm{Proc_F}\) is the set of finite processes from which observations are constructed.

\paragraph{Complete and Ready Simulation Observations.}

Complete and ready simulation are standard refinements of ordinary simulation. Complete simulation adds information about deadlocks, while ready simulation adds information about the actions that are initially enabled. We study what
happens when simulation observations are replaced by complete simulation and
ready simulation observations.

We define complete simulation and ready simulation \cite{GLABBEEK20013} uniformly using an index
\(X\in\{\mathrm{CS},\mathrm{RS}\}\).

\begin{definition}[\(X\)-simulation]
\label{def:X-simulation}
Let \(X\in\{\mathrm{CS},\mathrm{RS}\}\). A relation $R\subseteq \mathrm{Proc_F}\times \Proc$
is called an \emph{\(X\)-simulation} if whenever \((q,p)\in R\), the following
conditions hold:
\begin{enumerate}
\item[i.] if \(q\xrightarrow{a}_F q'\), then there exists \(p'\in\Proc\), such that $p\xrightarrow{a} p'$ and $(q',p')\in R$,
\item[ii.] the following \(X\)-condition holds: $\begin{cases}
I(q)=\varnothing \iff I(p)=\varnothing,
& \text{if } X=\mathrm{CS},\\[1mm]
I(q)=I(p),
& \text{if } X=\mathrm{RS}.
\end{cases}$
\end{enumerate}
\end{definition}

The case \(X=\mathrm{CS}\) corresponds to complete simulation, while the case
\(X=\mathrm{RS}\) corresponds to ready simulation. In the following, \(\sqsubseteq_{\mathrm{CS}}\) and \(\sqsubseteq_{\mathrm{RS}}\) denote complete and ready simulation, respectively.

\begin{definition}[\(X\)-simulation preorder]
\label{def:X-simulation-preorder}
For \(q\in\mathrm{Proc_F}\) and \(p\in\Proc\), we write
$q\sqsubseteq_X p$, if there exists an \(X\)-simulation relation $R\subseteq \mathrm{Proc_F}\times\Proc$, such that
$(q,p)\in R$.
\end{definition}

\begin{definition}[\(X\)-simulation observation]
\label{def:X-simulation-observation}
For \(q\in\mathrm{Proc_F}\), define $O_X(q)=\{\,p\in\Proc\mid q\sqsubseteq_X p\,\}$. Let $\mathcal F_X
=
\{\,O_X(q)\mid q\in\mathrm{Proc_F}\,\}$.
\end{definition}

We say that a family $\mathcal F\subseteq\mathcal P(\Proc)$ satisfies the \emph{covering property} if
$\bigcup_{U\in\mathcal F} U=\Proc$. 
Since every basis for a topology on $\Proc$ must satisfy the covering property, a family that fails to cover $\Proc$ cannot be a basis.

For ordinary simulation, the sets \(O(q)\) cover the process space and form a
basis. For complete simulation and ready simulation, this can fail. The reason
is that finite loop-free observations eventually terminate, while some processes
may continue forever, 
unable to reach a deadlock --- \emph{e.g.} a process $p$ with one transition $p \xrightarrow{a} p$. Therefore, such processes may
not satisfy any complete-simulation or ready-simulation observation. The next
result shows that these observation families fail already at the covering
property.

\begin{proposition}
\label{prop:FCS-FRS-not-cover}
Let \(L=\langle\Proc,\Sigma,\to\rangle\) be an LTS, and suppose there exists a process \(\hat{p}\in\Proc\setminus\mathrm{Proc}_F\), such that no path starting from \(\hat{p}\) reaches a deadlock state. Then, for each \(X\in\{\mathrm{CS},\mathrm{RS}\}\), the family $\mathcal F_X
=\{\,O_X(q)\mid q\in\mathrm{Proc}_F\,\}$ does not satisfy the covering property on \(\Proc\). Hence, neither
\(\mathcal F_{\mathrm{CS}}\) nor \(\mathcal F_{\mathrm{RS}}\) is a basis for a topology on \(\Proc\).
\end{proposition}

\begin{proof}[Proof idea]
The key point is that all finite loop-free observations eventually reach a deadlock, while some processes may run forever without reaching one. Both complete and ready simulations require information about whether states can terminate, or which actions are immediately available. Therefore, a process that never reaches a deadlock will fail to satisfy every complete- or ready-simulation observation. Hence, the families \(\mathcal F_{\mathrm{CS}}\) and \(\mathcal F_{\mathrm{RS}}\) do not
cover the process space, and so they cannot be bases.
\end{proof}

This result shows that stronger behavioural relations do not automatically produce suitable observation topologies. In the case of complete and ready simulation, finite loop-free observations are too limited to cover all processes. This motivates us to move to a bounded form of behavioural comparison, namely finite-depth bisimulation.

\subsection{Finite-Depth Bisimulation Observations}

Finite-depth bisimulation compares processes only up to a fixed finite depth \(k\). It keeps the two-sided matching condition of bisimulation, but applies it
only for a bounded number of steps. This allows us to retain finite observations while capturing information that ordinary simulation cannot.

\begin{definition}[\(k\)-bisimulation \cite{HennessyM85}]
\label{def:k-bisimulation}
For \(k\in\mathbb N\), the relation \(\sim_{k}\) between processes of LTSs is defined inductively as follows.

\noindent\textit{(Base case.)} For \(k=0\), $x \sim_{0} y$, for every \(x,y \in \Proc\).

\noindent\textit{(Inductive case.)} For \(k+1\), define \(x \sim_{k+1} y\), if the following hold:\begin{enumerate}[label=\textup{(\roman*)}, nosep, leftmargin=1.8em]
\item[i.] if $x\xrightarrow{a}x'$, then there exists \(y'\), such that $y\xrightarrow{a}y'$ and $x' \sim_{k} y'$,
\item[ii.] if $y\xrightarrow{a}y'$, then there exists \(x'\), such that $x\xrightarrow{a}x'$ and $x' \sim_{k} y'$.
\end{enumerate}
\end{definition}

Thus, \(x\sim_{k} y\) means that \(x\) and \(y\) cannot be distinguished by bisimulation observations of depth at most \(k\).

\begin{definition}
For \(q\in\Proc_F\) and \(k\in\mathbb N\), define $O^k_{\mathrm{bis}}(q)
=
\{\,p\in\Proc\mid q\sim_{k}p\,\}$.
The family of finite-depth bisimulation observations is $\mathcal B^{\mathrm{fin}}_{\mathrm{bis}}
=
\{\,O^k_{\mathrm{bis}}(q)\mid q\in\Proc_F,\ k\in\mathbb N\,\}$.
\end{definition}

\begin{proposition} 
\label{prop:bis-finite-depth-basis} 
The family \(\mathcal B^{\mathrm{fin}}_{\mathrm{bis}}\) is a basis for a topology $\tau^{\mathrm{fin}}_{\mathrm{bis}}$ on \(\Proc\). 
\end{proposition}

The topology \(\tau^{\mathrm{fin}}_{\mathrm{bis}}\) collects
bisimulation observations of all finite depths. Each observation is still
finite, but it can compare both directions of behaviour up to a bounded depth.
The resulting topology captures finite two-sided behavioural information that
is not available to the simulation topology.

The next result illustrates this difference through a deadlock. Ordinary
simulation observations cannot isolate the set of deadlocked processes, because
simulation only requires the observed process \(q\) to have its transitions
matched. If \(q\) has no transitions, then every process simulates it. In
contrast, finite-depth bisimulation also checks the reverse direction, and
therefore can detect whether a process has no outgoing transitions.

\begin{proposition}
\label{prop:D-bis-open-not-sim-open}
Assume that \(\Sigma\) is finite and contains at least one action, and that
\(\Proc\) contains at least one deadlocked process and at least one
non-deadlocked process. Let $D=\{\,p\in\Proc\mid I(p)=\varnothing\,\}$
be the set of deadlocked processes. Then, $D\in\tau^{\mathrm{fin}}_{\mathrm{bis}}$, and $
\Proc\setminus D\in\tau^{\mathrm{fin}}_{\mathrm{bis}}$, so \(D\) is clopen in \(\tau^{\mathrm{fin}}_{\mathrm{bis}}\). However, $
D\notin\tau_{\mathrm{sim}}$. Consequently, $\tau^{\mathrm{fin}}_{\mathrm{bis}}\neq\tau_{\mathrm{sim}}$.
\end{proposition}

\begin{proof}[Proof idea]
Finite-depth bisimulation can already identify deadlock at depth \(1\). Since
the process \(0\) has no outgoing transitions, a process is \(1\)-bisimilar to
\(0\) exactly when it also has no outgoing transitions.
Hence, $O^1_{\mathrm{bis}}(0)$ is exactly the set of deadlocked processes. The simulation topology cannot make this distinction. Since \(0\) has no outgoing transitions, every
process simulates \(0\), and therefore ordinary simulation cannot separate
deadlocked processes from non-deadlocked ones. Thus, deadlock is captured by
\(\tau^{\mathrm{fin}}_{\mathrm{bis}}\), but not by \(\tau_{\mathrm{sim}}\).
\end{proof}

This result shows that \(\tau^{\mathrm{fin}}_{\mathrm{bis}}\) captures finite
information that \(\tau_{\mathrm{sim}}\) cannot. In particular, deadlocks are observable in 
\(\tau^{\mathrm{fin}}_{\mathrm{bis}}\), but not in 
\(\tau_{\mathrm{sim}}\). Hence, finite-depth bisimulation does not merely refine the previous construction, as it generates a genuinely different topology.

\subsection{Modal Logic as a Topological Basis}
\label{sec: modal logic}

The Hennessy-Milner theorem~\cite{HennessyM85} tells us that bisimilarity is characterized by the Hennessy-Milner logic ($\HML$), in that two processes are bisimilar if and only if they satisfy exactly the same $\HML$ formulas. 
In other words, non-bisimilar processes can be separated by $\HML$ formulas, which is similar to the way that open sets separate points in a Hausdorff space. Furthermore, this correspondence extends between the preorders $\sqsubseteq_S, \sqsubseteq_{CS}, \sqsubseteq_{RS}$ and respective fragments of $\HML$~\cite{GLABBEEK20013,AcetoMFI19,DBLP:journals/iandc/GrooteV92}.
The similarities between formulas and open sets persist, as in this section we show that we can use fragments of $\HML$ as observation sets.

\begin{definition}[Fragments of 
$\HML$
]
We use the fragments of $\HML$ given by the following grammars:
\begin{align*}
\varphi_T \in \mathcal{L}_T &::= ~ \true ~ \mid ~ \false ~ \mid ~ \varphi_T\wedge \varphi_T~ \mid ~\varphi_T\vee \varphi_T~ \mid ~ \langle a \rangle\psi_T, 
&\psi_T  &::= ~ \true ~ \mid ~ \false ~ \mid ~  \langle a \rangle\psi_T \\
\varphi_S \in \mathcal{L}_S &::= ~ \true ~ \mid ~ \false ~ \mid ~ \varphi_S\wedge \varphi_S~ \mid ~\varphi_S\vee \varphi_S~ \mid ~ \langle a \rangle\varphi_S \\
\varphi_{BS} \in \HML &::= ~ \true ~ \mid ~ \false ~ \mid ~ \varphi_{BS}\wedge \varphi_{BS} ~ \mid ~\varphi_{BS}\vee \varphi_{BS}~ \mid ~ \langle a \rangle\varphi_{BS} ~ \mid ~ [a]\varphi_{BS},
\end{align*}
where $a \in \act$.
The semantics of the $\true, \false$ constants and $\land, \lor$ operators are the usual ones, and for the modal operators $[a]$ and $\langle a \rangle$, for each $p \in \Proc$, we define 
\begin{align*}
     &p\models \langle a\rangle\varphi \text{ iff there is 
 some } p \xrightarrow{a} q \text{ such that } q\models\varphi; \quad \text{and, dually,}\\
    &p\models [a]\varphi \text{ iff for all } p\xrightarrow{a} q \text{ it holds that } q\models\varphi  .
\end{align*}
The \emph{modal depth} of a formula is the nesting depth of its modalities. $\HML_k$ is the fragment of $\HML$ that only includes formulas of modal depth at most $k \in \Nat$. We use the notation $ [\![\varphi]\!] = \{ p \in \Proc \mid p \models \varphi \}$.
\end{definition}

The papers \cite{AcetoILS12,GrafS86a,SteffenI94} demonstrate that for each process $p$ and process preorder $\sqsubseteq \in \{\sqsubseteq_S, \sqsubseteq_{CS}, \sqsubseteq_{RS}\}$ (and others)
there exists a \emph{characteristic} formula in a corresponding fragment of $\HML$ that describes exactly the processes that are above $p$, with respect to $\sqsubseteq$.

\begin{theorem}
\label{modal-form}
    Let $\sqsubseteq$ be one of $\sqsubseteq_T^{fin}$, $\sqsubseteq_S$, $\sim_k$
    and let $\mathcal{L}$ be, respectively, $\mathcal{L}_T$, $\mathcal{L}_S$, or $\HML_k$.
    Then, for every $p \in \Proc_F$, there exists $\varphi \in \mathcal{L}$, such that for every $q \in \Proc$, 
    \( q \models \varphi \ \iff \ p \sqsubseteq q.\)
\end{theorem}

For each fragment $\mathcal{L}$ of $\HML$, we define the topology $\tau_{\mathcal{L}} = 
\langle 
\{ 
[\![\varphi]\!]
\mid \varphi \in \mathcal{L} \}
\rangle$. We note that $\{ [\![\varphi]\!] \mid \varphi \in \mathcal{L} \}$ is a basis, because the above fragments are closed with respect to conjunction and include $\true$.

\begin{corollary}
\label{cor:topol-equiv-logic}
    The topologies $\tau_O = \tau_{\mathcal{L}_T}$; $\tau_{sim} = \tau_{\mathcal{L}_S}$; and $\tau_{bis}^{fin} = \tau_{\HML}$.
\end{corollary}

\section{Conclusion}

This paper developed a topological framework for finite behavioural
observations and verification. Section~\ref{sec:trace-based topologies}
introduced full trace inclusion and its associated Alexandroff topology.
This section then showed that, on the trace
transition system, finite trace observations and finite-depth trace inclusion
both yield the Cantor topology. In contrast, full trace inclusion yields the
discrete topology.

For arbitrary labelled transition systems, Section~\ref{sec:process-topologies}
showed that finite positive trace observations induce \(\tau_O\), which
coincides with the topology generated by finite-depth trace inclusion. The same section established that finite loop-free simulation observations induce the strictly finer topology \(\tau_{\mathrm{sim}}\), showing that finite branching information can distinguish processes with the same finite traces.

Section~\ref{sec:verification} proved the general result that, whenever finite observations form a basis, the properties verifiable from those observations
are exactly the open sets of the induced topology. Applying this result,
we obtained multi-trace monitorability for
\(\tau_O\) and simulation monitorability for
\(\tau_{\mathrm{sim}}\), and showed that complete and ready simulation observations fail to form a basis. Therefore, complete and ready simulation cannot be verified using finite loop-free process observations on arbitrary LTSs. On the other hand, finite-depth bisimulation yields a finer topology in which deadlock is
observable. Finally, Subsection~\ref{sec: modal logic} provided a logical
presentation of the finite observations through modal formulas. 

Together, these results show that observation-induced topologies provide a
common framework for relating behavioural semantics, finite verification, and
monitorability. Future work includes studying the specialization preorders induced by these
observation topologies, in particular their relation to finitary simulation \cite{ACETO1997127,AcetoIngolfdottir1995},
and investigating whether finite observations for further behavioural relations,
such as prebisimulation, yield analogous topological characterisations.
Another interesting research avenue is to extend the set of finite observations to include finite-state processes that may loop.

\section*{Acknowledgment}

This research was funded by the Icelandic Research Fund (Ranns\'oknasj\'o\dh{}ur) through a doctoral student grant, grant no.~2511259.
Antonis Achilleos was also partly supported by the project `Hyperlogics: Expressiveness, Monitorability and Tools (H.-Lo)' (grant no.~2612260) of the Icelandic Research Fund.  
We sincerely thank Luca Aceto for his generous feedback, careful reading, and insightful suggestions, which have improved this work.

\bibliographystyle{eptcs}

\bibliography{generic.bib}

\begin{thebibliography}{10}
\providecommand{\bibitemdeclare}[2]{}
\providecommand{\surnamestart}{}
\providecommand{\surnameend}{}
\providecommand{\urlprefix}{Available at }
\providecommand{\url}[1]{\texttt{#1}}
\providecommand{\href}[2]{\texttt{#2}}
\providecommand{\urlalt}[2]{\href{#1}{#2}}
\providecommand{\doi}[1]{doi:\urlalt{https://doi.org/#1}{#1}}
\providecommand{\eprint}[1]{arXiv:\urlalt{https://arxiv.org/abs/#1}{#1}}
\providecommand{\bibinfo}[2]{#2}

\bibitemdeclare{misc}{abramsky2011domaintheorylogicobservable}
\bibitem{abramsky2011domaintheorylogicobservable}
\bibinfo{author}{Samson \surnamestart Abramsky\surnameend}
  (\bibinfo{year}{2011}): \emph{\bibinfo{title}{Domain Theory and the Logic of
  Observable Properties}}.
\newblock \eprint{1112.0347}.

\bibitemdeclare{article}{DBLP:journals/mscs/AbramskyV93}
\bibitem{DBLP:journals/mscs/AbramskyV93}
\bibinfo{author}{Samson \surnamestart Abramsky\surnameend} \&
  \bibinfo{author}{Steven \surnamestart Vickers\surnameend}
  (\bibinfo{year}{1993}): \emph{\bibinfo{title}{Quantales, Observational Logic
  and Process Semantics}}.
\newblock {\slshape \bibinfo{journal}{Math. Struct. Comput. Sci.}}
  \bibinfo{volume}{3}(\bibinfo{number}{2}), pp. \bibinfo{pages}{161--227},
  \doi{10.1017/S0960129500000189}.

\bibitemdeclare{inproceedings}{inbook2}
\bibitem{inbook2}
\bibinfo{author}{Luca \surnamestart Aceto\surnameend}, \bibinfo{author}{Antonis
  \surnamestart Achilleos\surnameend}, \bibinfo{author}{Adrian \surnamestart
  Francalanza\surnameend} \& \bibinfo{author}{Anna \surnamestart
  Ing{\'o}lfsd{\'o}ttir\surnameend} (\bibinfo{year}{2018}):
  \emph{\bibinfo{title}{A framework for parameterized monitorability}}.
\newblock In: {\slshape \bibinfo{booktitle}{International Conference on
  Foundations of Software Science and Computation Structures}},
  \bibinfo{organization}{Springer}, pp. \bibinfo{pages}{203--220},
  \doi{10.1007/978-3-319-89366-2_11}.

\bibitemdeclare{article}{10.1145/3290365}
\bibitem{10.1145/3290365}
\bibinfo{author}{Luca \surnamestart Aceto\surnameend}, \bibinfo{author}{Antonis
  \surnamestart Achilleos\surnameend}, \bibinfo{author}{Adrian \surnamestart
  Francalanza\surnameend}, \bibinfo{author}{Anna \surnamestart
  Ing\'{o}lfsd\'{o}ttir\surnameend} \& \bibinfo{author}{Karoliina \surnamestart
  Lehtinen\surnameend} (\bibinfo{year}{2019}): \emph{\bibinfo{title}{Adventures
  in monitorability: from branching to linear time and back again}}.
\newblock {\slshape \bibinfo{journal}{Proc. ACM Program. Lang.}}
  \bibinfo{volume}{3}(\bibinfo{number}{POPL}), \doi{10.1145/3290365}.

\bibitemdeclare{article}{Aceto2021}
\bibitem{Aceto2021}
\bibinfo{author}{Luca \surnamestart Aceto\surnameend}, \bibinfo{author}{Antonis
  \surnamestart Achilleos\surnameend}, \bibinfo{author}{Adrian \surnamestart
  Francalanza\surnameend}, \bibinfo{author}{Anna \surnamestart
  Ing{\'o}lfsd{\'o}ttir\surnameend} \& \bibinfo{author}{Karoliina \surnamestart
  Lehtinen\surnameend} (\bibinfo{year}{2021}): \emph{\bibinfo{title}{An
  operational guide to monitorability with applications to regular
  properties}}.
\newblock {\slshape \bibinfo{journal}{Software and Systems Modeling}}
  \bibinfo{volume}{20}, pp. \bibinfo{pages}{335--361},
  \doi{10.1007/s10270-020-00860-z}.

\bibitemdeclare{article}{AcetoMFI19}
\bibitem{AcetoMFI19}
\bibinfo{author}{Luca \surnamestart Aceto\surnameend}, \bibinfo{author}{Dario
  \surnamestart Della~Monica\surnameend}, \bibinfo{author}{Ignacio
  \surnamestart F{\'{a}}bregas\surnameend} \& \bibinfo{author}{Anna
  \surnamestart Ing{\'{o}}lfsd{\'{o}}ttir\surnameend} (\bibinfo{year}{2019}):
  \emph{\bibinfo{title}{When Are Prime Formulae Characteristic?}}
\newblock {\slshape \bibinfo{journal}{Theor. Comput. Sci.}}
  \bibinfo{volume}{777}, pp. \bibinfo{pages}{3--31},
  \doi{10.1016/J.TCS.2018.12.004}.

\bibitemdeclare{techreport}{AcetoIngolfdottir1995}
\bibitem{AcetoIngolfdottir1995}
\bibinfo{author}{Luca \surnamestart Aceto\surnameend} \& \bibinfo{author}{Anna
  \surnamestart Ingolfsdottir\surnameend} (\bibinfo{year}{1995}):
  \emph{\bibinfo{title}{On the Finitary Bisimulation}}.
\newblock \bibinfo{type}{Technical Report} \bibinfo{number}{59},
  \bibinfo{institution}{BRICS, Department of Computer Science, Aarhus
  University}, \doi{10.7146/brics.v2i59.19960}.

\bibitemdeclare{book}{Aceto07}
\bibitem{Aceto07}
\bibinfo{author}{Luca \surnamestart Aceto\surnameend}, \bibinfo{author}{Anna
  \surnamestart Ing{\'o}lfsd{\'o}ttir\surnameend},
  \bibinfo{author}{Kim~Guldstrand \surnamestart Larsen\surnameend} \&
  \bibinfo{author}{Jiri \surnamestart Srba\surnameend} (\bibinfo{year}{2007}):
  \emph{\bibinfo{title}{Reactive systems: modelling, specification and
  verification}}.
\newblock \bibinfo{publisher}{Cambridge University Press},
  \doi{10.1017/CBO9780511814105}.

\bibitemdeclare{article}{AcetoILS12}
\bibitem{AcetoILS12}
\bibinfo{author}{Luca \surnamestart Aceto\surnameend}, \bibinfo{author}{Anna
  \surnamestart Ing{\'{o}}lfsd{\'{o}}ttir\surnameend},
  \bibinfo{author}{Paul~Blain \surnamestart Levy\surnameend} \&
  \bibinfo{author}{Joshua \surnamestart Sack\surnameend}
  (\bibinfo{year}{2012}): \emph{\bibinfo{title}{Characteristic formulae for
  fixed-point semantics: a general framework}}.
\newblock {\slshape \bibinfo{journal}{Math. Struct. Comput. Sci.}}
  \bibinfo{volume}{22}(\bibinfo{number}{2}), pp. \bibinfo{pages}{125--173},
  \doi{10.1017/S0960129511000375}.

\bibitemdeclare{article}{ACETO1997127}
\bibitem{ACETO1997127}
\bibinfo{author}{Luca \surnamestart Aceto\surnameend} \& \bibinfo{author}{Anna
  \surnamestart Ingólfsdóttir\surnameend} (\bibinfo{year}{1997}):
  \emph{\bibinfo{title}{A characterization of finitary bisimulation}}.
\newblock {\slshape \bibinfo{journal}{Information Processing Letters}}
  \bibinfo{volume}{64}(\bibinfo{number}{3}), pp. \bibinfo{pages}{127--134},
  \doi{10.1016/S0020-0190(97)00163-4}.
\newblock
  \urlprefix\url{https://www.sciencedirect.com/science/article/pii/S0020019097001634}.

\bibitemdeclare{incollection}{inbook}
\bibitem{inbook}
\bibinfo{author}{Ezio \surnamestart Bartocci\surnameend},
  \bibinfo{author}{Yli{\`e}s \surnamestart Falcone\surnameend},
  \bibinfo{author}{Adrian \surnamestart Francalanza\surnameend} \&
  \bibinfo{author}{Giles \surnamestart Reger\surnameend}
  (\bibinfo{year}{2018}): \emph{\bibinfo{title}{Introduction to runtime
  verification}}.
\newblock In: {\slshape \bibinfo{booktitle}{Lectures on Runtime Verification:
  Introductory and Advanced Topics}}, \bibinfo{publisher}{Springer}, pp.
  \bibinfo{pages}{1--33}, \doi{10.1007/978-3-319-75632-5}.

\bibitemdeclare{article}{c9c054ed-1535-3e93-8569-48e7acb4c9c7}
\bibitem{c9c054ed-1535-3e93-8569-48e7acb4c9c7}
\bibinfo{author}{I.~\surnamestart Berlin\surnameend} (\bibinfo{year}{1938}):
  \emph{\bibinfo{title}{Verification}}.
\newblock {\slshape \bibinfo{journal}{Proceedings of the Aristotelian Society}}
  \bibinfo{volume}{39}, pp. \bibinfo{pages}{225--248},
  \doi{10.1093/aristotelian/39.1.225}.
\newblock \urlprefix\url{http://www.jstor.org/stable/4544328}.

\bibitemdeclare{misc}{CamerloDagnino2026}
\bibitem{CamerloDagnino2026}
\bibinfo{author}{Riccardo \surnamestart Camerlo\surnameend} \&
  \bibinfo{author}{Francesco \surnamestart Dagnino\surnameend}
  (\bibinfo{year}{2026}): \emph{\bibinfo{title}{The Complexity of Being
  Monitorable}}, \doi{10.48550/arXiv.2601.04256}.

\bibitemdeclare{inproceedings}{10.1007/3-540-53479}
\bibitem{10.1007/3-540-53479}
\bibinfo{author}{Rocco \surnamestart De~Nicola\surnameend} \&
  \bibinfo{author}{Frits \surnamestart Vaandrager\surnameend}
  (\bibinfo{year}{1990}): \emph{\bibinfo{title}{Action versus state based
  logics for transition systems}}.
\newblock In \bibinfo{editor}{Ir{\`e}ne \surnamestart Guessarian\surnameend},
  editor: {\slshape \bibinfo{booktitle}{Semantics of Systems of Concurrent
  Processes}}, \bibinfo{publisher}{Springer Berlin Heidelberg},
  \bibinfo{address}{Berlin, Heidelberg}, pp. \bibinfo{pages}{407--419},
  \doi{10.1007/3-540-53479-2_17}.

\bibitemdeclare{article}{DIEKERT201429}
\bibitem{DIEKERT201429}
\bibinfo{author}{Volker \surnamestart Diekert\surnameend} \&
  \bibinfo{author}{Martin \surnamestart Leucker\surnameend}
  (\bibinfo{year}{2014}): \emph{\bibinfo{title}{Topology, monitorable
  properties and runtime verification}}.
\newblock {\slshape \bibinfo{journal}{Theoretical Computer Science}}
  \bibinfo{volume}{537}, pp. \bibinfo{pages}{29--41},
  \doi{10.1016/j.tcs.2014.02.052}.
\newblock \bibinfo{note}{Theoretical Aspects of Computing (ICTAC 2011)}.

\bibitemdeclare{incollection}{GLABBEEK20013}
\bibitem{GLABBEEK20013}
\bibinfo{author}{Rob~J. \surnamestart van Glabbeek\surnameend}
  (\bibinfo{year}{2001}): \emph{\bibinfo{title}{{The Linear Time--Branching
  Time Spectrum I}}}.
\newblock In \bibinfo{editor}{Jan~A. \surnamestart Bergstra\surnameend},
  \bibinfo{editor}{Alban \surnamestart Ponse\surnameend} \&
  \bibinfo{editor}{Scott~A. \surnamestart Smolka\surnameend}, editors:
  {\slshape \bibinfo{booktitle}{Handbook of Process Algebra}},
  \bibinfo{publisher}{North-Holland / Elsevier}, pp. \bibinfo{pages}{3--99},
  \doi{10.1016/B978-044482830-9/50019-9}.

\bibitemdeclare{article}{GrafS86a}
\bibitem{GrafS86a}
\bibinfo{author}{Susanne \surnamestart Graf\surnameend} \&
  \bibinfo{author}{Joseph \surnamestart Sifakis\surnameend}
  (\bibinfo{year}{1986}): \emph{\bibinfo{title}{A Modal Characterization of
  Observational Congruence on Finite Terms of {CCS}}}.
\newblock {\slshape \bibinfo{journal}{Inf. Control.}}
  \bibinfo{volume}{68}(\bibinfo{number}{1-3}), pp. \bibinfo{pages}{125--145},
  \doi{10.1016/S0019-9958(86)80031-6}.

\bibitemdeclare{article}{DBLP:journals/iandc/GrooteV92}
\bibitem{DBLP:journals/iandc/GrooteV92}
\bibinfo{author}{Jan~Friso \surnamestart Groote\surnameend} \&
  \bibinfo{author}{Frits~W. \surnamestart Vaandrager\surnameend}
  (\bibinfo{year}{1992}): \emph{\bibinfo{title}{Structured Operational
  Semantics and Bisimulation as a Congruence}}.
\newblock {\slshape \bibinfo{journal}{Inf. Comput.}}
  \bibinfo{volume}{100}(\bibinfo{number}{2}), pp. \bibinfo{pages}{202--260},
  \doi{10.1016/0890-5401(92)90013-6}.

\bibitemdeclare{inproceedings}{10.1007/978-3-642-82453}
\bibitem{10.1007/978-3-642-82453}
\bibinfo{author}{D.~\surnamestart Harel\surnameend} \&
  \bibinfo{author}{A.~\surnamestart Pnueli\surnameend} (\bibinfo{year}{1985}):
  \emph{\bibinfo{title}{On the Development of Reactive Systems}}.
\newblock In \bibinfo{editor}{Krzysztof~R. \surnamestart Apt\surnameend},
  editor: {\slshape \bibinfo{booktitle}{Logics and Models of Concurrent
  Systems}}, \bibinfo{publisher}{Springer Berlin Heidelberg},
  \bibinfo{address}{Berlin, Heidelberg}, pp. \bibinfo{pages}{477--498},
  \doi{10.1007/978-3-642-82453-1_17}.

\bibitemdeclare{article}{HennessyM85}
\bibitem{HennessyM85}
\bibinfo{author}{Matthew \surnamestart Hennessy\surnameend} \&
  \bibinfo{author}{Robin \surnamestart Milner\surnameend}
  (\bibinfo{year}{1985}): \emph{\bibinfo{title}{Algebraic Laws for
  Nondeterminism and Concurrency}}.
\newblock {\slshape \bibinfo{journal}{J. {ACM}}}
  \bibinfo{volume}{32}(\bibinfo{number}{1}), pp. \bibinfo{pages}{137--161},
  \doi{10.1145/2455.2460}.

\bibitemdeclare{article}{lmcs:2271}
\bibitem{lmcs:2271}
\bibinfo{author}{Michael \surnamestart Huth\surnameend} (\bibinfo{year}{2005}):
  \emph{\bibinfo{title}{Labelled transition systems as a Stone space}}.
\newblock {\slshape \bibinfo{journal}{Logical Methods in Computer Science}}
  \bibinfo{volume}{Volume 1, Issue 1}:\bibinfo{eid}{6},
  \doi{10.2168/LMCS-1(1:1)2005}.
\newblock \urlprefix\url{https://lmcs.episciences.org/2271}.

\bibitemdeclare{book}{Kechris1995}
\bibitem{Kechris1995}
\bibinfo{author}{A.~S. \surnamestart Kechris\surnameend}
  (\bibinfo{year}{1995}): \emph{\bibinfo{title}{Classical Descriptive Set
  Theory}}.
\newblock \bibinfo{publisher}{Springer-Verlag}, \bibinfo{address}{New York},
  \doi{10.1007/978-1-4612-4190-4}.

\bibitemdeclare{article}{KLEIJNEN1995145}
\bibitem{KLEIJNEN1995145}
\bibinfo{author}{Jack~P.C. \surnamestart Kleijnen\surnameend}
  (\bibinfo{year}{1995}): \emph{\bibinfo{title}{Verification and validation of
  simulation models}}.
\newblock {\slshape \bibinfo{journal}{European Journal of Operational
  Research}} \bibinfo{volume}{82}(\bibinfo{number}{1}), pp.
  \bibinfo{pages}{145--162}, \doi{10.1016/0377-2217(94)00016-6}.
\newblock
  \urlprefix\url{https://www.sciencedirect.com/science/article/pii/0377221794000166}.

\bibitemdeclare{article}{LEUCKER2009293}
\bibitem{LEUCKER2009293}
\bibinfo{author}{Martin \surnamestart Leucker\surnameend} \&
  \bibinfo{author}{Christian \surnamestart Schallhart\surnameend}
  (\bibinfo{year}{2009}): \emph{\bibinfo{title}{A brief account of runtime
  verification}}.
\newblock {\slshape \bibinfo{journal}{The Journal of Logic and Algebraic
  Programming}} \bibinfo{volume}{78}(\bibinfo{number}{5}), pp.
  \bibinfo{pages}{293--303}, \doi{10.1016/j.jlap.2008.08.004}.
\newblock
  \urlprefix\url{https://www.sciencedirect.com/science/article/pii/S1567832608000775}.
\newblock \bibinfo{note}{The 1st Workshop on Formal Languages and Analysis of
  Contract-Oriented Software (FLACOS’07)}.

\bibitemdeclare{book}{4cd7021d283c4edab373d42f1d69a14c}
\bibitem{4cd7021d283c4edab373d42f1d69a14c}
\bibinfo{author}{R.~\surnamestart Milner\surnameend} (\bibinfo{year}{1989}):
  \emph{\bibinfo{title}{Communication and Concurrency}}.
\newblock \bibinfo{publisher}{Prentice-Hall, Inc.}

\bibitemdeclare{book}{Moschovakis2009}
\bibitem{Moschovakis2009}
\bibinfo{author}{Y.~N. \surnamestart Moschovakis\surnameend}
  (\bibinfo{year}{2009}): \emph{\bibinfo{title}{Descriptive Set Theory}},
  \bibinfo{edition}{second} edition.
\newblock {\slshape \bibinfo{series}{Mathematical Surveys and Monographs}}
  \bibinfo{volume}{155}, \bibinfo{publisher}{American Mathematical Society},
  \doi{10.1090/surv/155/08}.

\bibitemdeclare{inproceedings}{10.1007/3-540-44685}
\bibitem{10.1007/3-540-44685}
\bibinfo{author}{Ji{\v{r}}{\'i} \surnamestart Srba\surnameend}
  (\bibinfo{year}{2001}): \emph{\bibinfo{title}{On the Power of Labels in
  Transition Systems}}.
\newblock In \bibinfo{editor}{Kim~G. \surnamestart Larsen\surnameend} \&
  \bibinfo{editor}{Mogens \surnamestart Nielsen\surnameend}, editors: {\slshape
  \bibinfo{booktitle}{CONCUR 2001 --- Concurrency Theory}},
  \bibinfo{publisher}{Springer Berlin Heidelberg}, \bibinfo{address}{Berlin,
  Heidelberg}, pp. \bibinfo{pages}{277--291}, \doi{10.1007/3-540-44685-0_19}.

\bibitemdeclare{article}{SteffenI94}
\bibitem{SteffenI94}
\bibinfo{author}{Bernhard \surnamestart Steffen\surnameend} \&
  \bibinfo{author}{Anna \surnamestart Ing{\'{o}}lfsd{\'{o}}ttir\surnameend}
  (\bibinfo{year}{1994}): \emph{\bibinfo{title}{Characteristic Formulae for
  Processes with Divergence}}.
\newblock {\slshape \bibinfo{journal}{Inf. Comput.}}
  \bibinfo{volume}{110}(\bibinfo{number}{1}), pp. \bibinfo{pages}{149--163},
  \doi{10.1006/INCO.1994.1028}.

\bibitemdeclare{techreport}{osti835920}
\bibitem{osti835920}
\bibinfo{author}{B~H \surnamestart Thacker\surnameend}, \bibinfo{author}{S~W
  \surnamestart Doebling\surnameend}, \bibinfo{author}{F~M \surnamestart
  Hemez\surnameend}, \bibinfo{author}{M~C \surnamestart Anderson\surnameend},
  \bibinfo{author}{J~E \surnamestart Pepin\surnameend} \& \bibinfo{author}{E~A
  \surnamestart Rodriguez\surnameend} (\bibinfo{year}{2004}):
  \emph{\bibinfo{title}{Concepts of Model Verification and Validation}}.
\newblock \bibinfo{type}{Technical Report}, \bibinfo{institution}{Los Alamos
  National Lab., Los Alamos, NM (US)}, \doi{10.2172/835920}.
\newblock \urlprefix\url{https://www.osti.gov/biblio/835920}.

\bibitemdeclare{incollection}{Tretmans2008}
\bibitem{Tretmans2008}
\bibinfo{author}{Jan \surnamestart Tretmans\surnameend} (\bibinfo{year}{2008}):
  \emph{\bibinfo{title}{Model Based Testing with Labelled Transition Systems}}.
\newblock In \bibinfo{editor}{Robert~M. \surnamestart Hierons\surnameend},
  \bibinfo{editor}{Jonathan~P. \surnamestart Bowen\surnameend} \&
  \bibinfo{editor}{Mark \surnamestart Harman\surnameend}, editors: {\slshape
  \bibinfo{booktitle}{Formal Methods and Testing: An Outcome of the FORTEST
  Network, Revised Selected Papers}}, \bibinfo{publisher}{Springer Berlin
  Heidelberg}, \bibinfo{address}{Berlin, Heidelberg}, pp.
  \bibinfo{pages}{1--38}, \doi{10.1007/978-3-540-78917-8_1}.

\bibitemdeclare{book}{vickers1988topology}
\bibitem{vickers1988topology}
\bibinfo{author}{Steven \surnamestart Vickers\surnameend}
  (\bibinfo{year}{1988}): \emph{\bibinfo{title}{Topology via Logic}}.
\newblock {\slshape \bibinfo{series}{Cambridge Tracts in Theoretical Computer
  Science}}~\bibinfo{volume}{5}, \bibinfo{publisher}{Cambridge University
  Press}.

\end{thebibliography}

\end{document}